\documentclass{article}
\usepackage{amssymb}

\usepackage{amsmath}

\newtheorem{theorem}{Theorem}

\newtheorem{proposition}[theorem]{Proposition}
\newtheorem{remark}[theorem]{Remark}

\newenvironment{proof}[1][Proof]{\textbf{#1.} }{\ \rule{0.5em}{0.5em}}
\input{tcilatex}

\begin{document}

\title{On the Exponential Decay of the n-Point Correlation Functions and the
Analyticity of the Pressure. }
\author{Assane Lo}
\date{April 1st, 2007}
\maketitle

\begin{abstract}
The goal of this paper is to provide estimates leading to a direct proof of
the exponential decay of the n-point correlation functions for\ certain
unbounded models of Kac type. The methods are based on estimating higher
order derivatives of the solution of the Witten Laplacian equation on one
forms associated with the hamiltonian of the system. We also provide a
formula for the Taylor coefficients of the pressure that is suitable for a
direct proof the analyticity.
\end{abstract}

\section{Introduction}

In recent publications [66] we have given a generalization to the higher
dimensional case of the exponential decay of the two-point correlation
functions for models of Kac type. In this paper, we shall establish a weak
exponential decay of the n-point correlation functions, and provide an exact
formula suitable for a direct proof the analyticity of the pressure.

Let $\Lambda $ be a finite subset of $\mathbb{Z}^{d}$, and consider a
Hamiltonian $\Phi $ of the phase space $\mathbb{R}^{\Lambda }.$ We shall
focus on the case where $\Phi =\Phi _{\Lambda }$ is given by 
\begin{equation}
\Phi _{\Lambda }(x)=\frac{x^{2}}{2}+\Psi (x),
\end{equation}%
under suitable assumptions on $\Psi .$

Recall that if $\left\langle f\right\rangle $ denote the mean value of $f$
with respect to the\ Gibbs measure 
\begin{equation*}
e^{-\Phi (x)}dx,
\end{equation*}%
the covariance of two functions $g$ and $h$ is defined by 
\begin{equation}
\mathbf{cov}(g,h)=\left\langle (g-\left\langle g\right\rangle
)(h-\left\langle h\right\rangle )\right\rangle .
\end{equation}%
If one wants to have an expression of the covariance in the form 
\begin{equation}
\mathbf{cov}(g,h)=\left\langle \mathbf{\nabla }h\cdot \mathbf{w}%
\right\rangle _{L^{2}(\mathbb{R}^{n},\mathbb{R}^{n};e^{-\Phi }dx)},
\end{equation}%
for a suitable vector field $\mathbf{w}$ we get, after observing that $%
\mathbf{\nabla }h=\mathbf{\nabla }(h-\left\langle h\right\rangle ),$ and
integrating by parts,%
\begin{equation}
\mathbf{cov}(g,h)=\int (h-\left\langle h\right\rangle )(\mathbf{\nabla }\Phi
-\mathbf{\nabla })\cdot \mathbf{w}e^{-\Phi (x)}dx.
\end{equation}%
(Here we have assumed that $g$ and $h$ are functions of polynomial growth ).

This leads to the question of solving the equation%
\begin{equation}
g-\left\langle g\right\rangle =(\mathbf{\nabla }\Phi -\mathbf{\nabla })\cdot 
\mathbf{w.}
\end{equation}%
Now, trying to solve this above equation with $\mathbf{w}=\mathbf{\nabla }f,$
we obtain the equation%
\begin{equation}
\left. 
\begin{array}{c}
g-\left\langle g\right\rangle =\left( -\mathbf{\Delta }+\mathbf{\nabla }\Phi
\cdot \mathbf{\nabla }\right) f \\ 
\left\langle f\right\rangle =0.%
\end{array}%
\right\}
\end{equation}%
The existence and smoothness of the solution of this equation were
established in [8] (see also [66]) under certain assumptions on $\Phi .$ Now
taking gradient on both sides of $(6),$ we get%
\begin{equation}
\mathbf{\nabla }g=\left[ \left( -\mathbf{\Delta }+\mathbf{\nabla }\Phi \cdot 
\mathbf{\nabla }\right) \otimes Id\mathbf{+Hess}\Phi \right] \mathbf{\nabla }%
f.
\end{equation}%
We then obtain the emergence of two differential operators:%
\begin{equation}
A_{\Phi }^{(0)}:=-\mathbf{\Delta }+\mathbf{\nabla }\Phi \cdot \mathbf{\nabla 
}
\end{equation}%
and%
\begin{equation}
A_{\Phi }^{(1)}:=A_{\Phi }^{(0)}\otimes Id\mathbf{+Hess}\Phi .
\end{equation}%
Here the tensor notation means that $A_{\Phi }^{(0)}$ acts diagonally on the
vector field solution to produce a system of equations.

Thus 
\begin{equation}
\mathbf{cov}(g,h)=\int \left( A_{\Phi }^{(1)^{-1}}\mathbf{\nabla }g\cdot 
\mathbf{\nabla }h\right) e^{-\Phi (x)}dx.
\end{equation}%
The operators $A_{\Phi }^{(0)}$ and $A_{\Phi }^{(1)}$ are called the
Helffer-Sj\"{o}strand's operators. These are unbounded operators acting on
the weighted spaces 
\begin{equation*}
L^{2}(\mathbb{R}^{\Lambda },e^{-\Phi }dx)\;\text{and\ }L^{2}(\mathbb{R}%
^{\Lambda },\mathbb{R}^{\Lambda },e^{-\Phi }dx)
\end{equation*}%
respectively.

The formula (10) was introduced by Helffer and Sj\"{o}strand and in some
sense is a generalization of Brascamp-Lieb inequality as already pointed out
in [1].

The unitary transformation%
\begin{eqnarray*}
U_{\Phi } &:&L^{2}(\mathbb{R}^{\Lambda })\rightarrow L^{2}(\mathbb{R}%
^{\Lambda },e^{-\Phi }dx) \\
u &\longmapsto &e^{\frac{\Phi }{2}}u
\end{eqnarray*}%
will allow us to work with the unweighted spaces $L^{2}(\mathbb{R}^{\Lambda
})$ and $L^{2}(\mathbb{R}^{\Lambda },\mathbb{R}^{\Lambda })$ by converting
the operators $A_{\Phi }^{(0)}$ and $A_{\Phi }^{(1)}$ into equivalent
operators 
\begin{equation}
\mathbf{W}_{\Phi }^{\left( 0\right) }\mathbf{=-\Delta +}\frac{\left| \mathbf{%
\nabla }\Phi \right| ^{2}}{4}-\frac{\mathbf{\Delta }\Phi }{2}
\end{equation}%
and%
\begin{equation}
\mathbf{W}_{\Phi }^{\left( 1\right) }=\left( \mathbf{-\Delta +}\frac{\left| 
\mathbf{\nabla }\Phi \right| ^{2}}{4}-\frac{\mathbf{\Delta }\Phi }{2}\right)
\otimes \mathbf{I}+\mathbf{Hess}\Phi .
\end{equation}%
respectively.

The equivalence can be seen by observing that 
\begin{equation}
W_{\Phi }^{(.)}=e^{-\Phi /2}\circ A_{\Phi }^{(.)}\circ e^{\Phi /2}.
\end{equation}%
The operators $\mathbf{W}_{\Phi }^{\left( 0\right) }$ and $\mathbf{W}_{\Phi
}^{\left( 1\right) }$ are unbounded operators acting on 
\begin{equation*}
L^{2}(\mathbb{R}^{\Lambda })\text{ \ and \ }L^{2}(\mathbb{R}^{\Lambda },%
\mathbb{R}^{\Lambda })
\end{equation*}%
respectively. These are in fact, the euclidean versions of the Laplacians on
zero and one forms respectively, already introduced by E. Witten [18] in the
context Morse theory.

The equivalence between the operators $A_{\Phi }^{(.)}$ and Witten's
Laplacians was first observed by J. Sj\"{o}strand [13] in 1996.

\section{Higher Order Exponential Estimates}

We shall consider a Hamiltonian of the form 
\begin{equation*}
\Phi (x)=\Phi _{\Lambda }(x)=\frac{x^{2}}{2}+\Psi (x),\;\ \ \ \ \;x\in 
\mathbb{R}^{\Lambda }.
\end{equation*}%
where%
\begin{equation}
\left| \partial ^{\alpha }\mathbf{\nabla }\Psi \right| \leq C_{\alpha },\;\
\ \ \ \ \forall \alpha \in \mathbb{N}^{\left| \Lambda \right| }.
\end{equation}%
$g$ will denote a smooth function on $\mathbb{R}^{\Gamma }$ with lattice
support $S_{g}=\Gamma \left( \varsubsetneqq \Lambda \right) .$ We shall
identify $g$ with $\tilde{g}$ defined on $\mathbb{R}^{\Lambda }$ and shall
assume that 
\begin{equation}
\left| \partial ^{\alpha }\mathbf{\nabla }g\right| \leq C_{\alpha
}\;\;\;\;\;\;\;\;\forall \alpha \in \mathbb{N}^{\left| \Gamma \right| }.
\end{equation}%
As in [66] , we shall momentarily assume that $\Psi $ is compactly supported
in $\mathbb{R}^{\Lambda }$ and $g$ is compactly supported in $\mathbb{R}%
^{\Gamma }$ but these assumptions will be relaxed later on.

Let $M$ be the diagonal matrix%
\begin{equation*}
M=\left( \delta _{ij}\rho (i)\right) _{i,j\in \Lambda }
\end{equation*}%
where $\rho $ is a weight function on $\Lambda $ satisfying \ 
\begin{equation}
e^{-\lambda }\leq \frac{\rho \left( i\right) }{\rho (j)}\leq e^{\lambda }%
\text{, \ \ if }i\sim j\text{\ \ for some }\lambda >0.
\end{equation}%
Assume also that for every $M$ as above, there exists $\delta _{o}\in (0,1)$
such that%
\begin{equation}
\left\langle M^{-1}\mathbf{Hess}\Phi (x)Ma,a\right\rangle \geq \delta
_{o}a^{2},\text{ \ \ \ \ }\forall x\in \mathbb{R}^{\Lambda },\forall a\in 
\mathbb{R}^{\Lambda }.
\end{equation}

For instance, the $d-$dimensional nearest neighbor Kac model 
\begin{equation*}
\Phi _{\Lambda }(x)=\frac{x^{2}}{2}-2\sum_{i\sim j}\ln \cosh \left[ \sqrt{%
\frac{\nu }{2}}\left( x_{i}+x_{j}\right) \right] .
\end{equation*}%
satisfies this assumption for $\nu $ small enough. See [66] for details.

The following theorem has been proved in [66]:

\begin{theorem}[{A. Lo [66]}]
\textit{Let }$g$\textit{\ be\ a smooth function with compact support on }$%
\mathbb{R}^{\Gamma }$ \textit{satisfying}%
\begin{equation}
\left| \partial ^{\alpha }\mathbf{\nabla }g\right| \leq C_{\alpha
}\;\;\;\;\;\;\;\forall \alpha \in \mathbb{N}^{\left| \Gamma \right| }
\end{equation}%
and $\Phi $\textit{\ is as above. If }$f$ is \textit{the unique C}$^{\infty
}-$\textit{solution of the equation }%
\begin{equation*}
\left\{ 
\begin{tabular}{l}
$-\mathbf{\Delta }f+\mathbf{\nabla }\Phi \cdot \mathbf{\nabla }%
f=g-\left\langle g\right\rangle $ \\ 
$\left\langle f\right\rangle _{L^{2}(\mu )}=0,$%
\end{tabular}%
\right.
\end{equation*}%
\textit{then}%
\begin{equation*}
\sum_{i\in \Lambda }f_{x_{i}}^{2}(x)e^{2\kappa d(i,S_{g})}\leq C\;\;\ \
\forall x\in \mathbb{R}^{\Lambda }.
\end{equation*}%
$\kappa $ and $C$ \textit{are positive constants. }$C$\textit{\ could
possibly depend on the size of the support of }$g$\textit{\ but does not
depend on }$\Lambda $\textit{\ and }$f.$
\end{theorem}

We now propose to generalize this theorem to higher order derivatives.

\begin{proposition}
\textit{If in addition to the assumptions of theorem 1, }$\Phi $\textit{\
satisfies the following growth condition: for }$\kappa >0$ \textit{as above,}%
\begin{equation}
\sum_{j,i_{1},...,i_{k}\in \Lambda }\Phi
_{x_{j}x_{i_{1}}...x_{i_{k}}}^{2}(x)e^{2\kappa d(\left\{
i_{1},...,i_{k}\right\} ,S_{g})}\leq C_{k}\;\;\;\;\;\;\;\forall x\in \mathbb{%
R}^{\Lambda },\text{\textit{\ for }}k\geq 2
\end{equation}%
\textit{for some }$C_{k}>0$\textit{\ not dependent on }$\Lambda $ and $f$, 
\textit{then for any }$k\geq 1,$ $f$\textit{\ satisfies }%
\begin{equation}
\sum_{i_{1},...,i_{k}\in \Lambda }f_{x_{i_{1}}...x_{i_{k}}}^{2}(x)e^{2\kappa
d(\left\{ i_{1},...,i_{k}\right\} ,S_{g})}\leq C_{k,g}\;\;\;\;\;\;\;\forall
x\in \mathbb{R}^{\Lambda }
\end{equation}%
\textit{where }$C_{k,g}>0$\textit{\ is a constant that depends on the size
of the support of }$g$ \textit{but not on }$\Lambda $ and $f.$
\end{proposition}

\textbf{Proof.}

The case $k=1$ being theorem 1, we assume for induction that the result is
true when $k$ is replaced by $\hat{k}<k$ with $\hat{k}\geq 2.$

For $k\geq 2($see [8] for details), we have%
\begin{eqnarray*}
\left\langle \mathbf{\nabla }^{k}g,t_{1}\otimes ...\otimes
t_{k}\right\rangle &=&\left( \mathbf{\nabla }\Phi \cdot \mathbf{\nabla }-%
\mathbf{\Delta }\right) \left\langle \mathbf{\nabla }^{k}f,t_{1}\otimes
...\otimes t_{k}\right\rangle \\
&&+\sum_{j=1}^{k}\left\langle \mathbf{\nabla }^{k}f,t_{1}\otimes ...\otimes 
\mathbf{Hess}\Phi t_{j}\otimes ...\otimes t_{k}\right\rangle \\
&&+\sum\limits_{\substack{ A\cup B=\{1,...,k\},A\cap B=\emptyset  \\ \#B\leq
k-2}}\left\langle t_{A}(\partial _{x})\mathbf{\nabla }\Phi ,t_{B}(\partial
_{x})\mathbf{\nabla }f\right\rangle .
\end{eqnarray*}%
In the right hand side of this last above equality, we have used the notation%
\begin{equation*}
t_{J}(\partial _{x})u:=\left\langle \mathbf{\nabla }^{\#J}u,t_{1}\otimes
...\otimes t_{\#J}\right\rangle .
\end{equation*}%
Now fix $i_{2},...,i_{k}\in \Lambda .$ Because $\mathbf{\nabla }%
^{k}f(x)\rightarrow 0$ as $\left| x\right| \rightarrow \infty $ (see [66]),
we consider $x_{o}\in \mathbb{R}^{\Lambda }$ that maximizes%
\begin{equation*}
x\longmapsto \sum_{i_{1}}f_{x_{i_{1}}...x_{i_{k}}}^{2}\rho
^{2}(i_{1},...,i_{k})
\end{equation*}%
where%
\begin{equation*}
\rho (i_{1},...,i_{k})=e^{\kappa d(\left\{ i_{1},...,i_{k}\right\} ,S_{g})}.
\end{equation*}%
Observe here that $x_{o}$ could possibly depend on $i_{2},...,i_{k}\in
\Lambda .$

Choose 
\begin{equation*}
t_{1}=\left( \rho (i_{1},...,i_{k})f_{x_{i_{1}}...x_{i_{k}}}(x_{o})\right)
_{i_{1}\in \Lambda }
\end{equation*}%
and 
\begin{equation*}
t_{j}=e_{i_{j}}\;\;\;\text{ if \ \ \ }j=2,...,k
\end{equation*}%
Let $M_{1}$ be the diagonal matrix%
\begin{equation*}
M_{1}=\left( \delta _{si_{1}}\rho (i_{1},...,i_{k})\right) _{si_{1}}
\end{equation*}%
and\ 
\begin{equation}
M_{j}=\mathbf{I}\text{ \ if }j\neq 1
\end{equation}%
in particular, we have%
\begin{eqnarray*}
&&\left\langle \mathbf{\nabla }^{k}g,M_{1}t_{1}\otimes ...\otimes
M_{k}t_{k}\right\rangle \\
&=&\left( \mathbf{\nabla }\Phi \cdot \mathbf{\nabla }-\mathbf{\Delta }%
\right) \left\langle \mathbf{\nabla }^{k}f,M_{1}t_{1}\otimes ...\otimes
M_{k}t_{k}\right\rangle \\
&&+\sum_{j=1}^{k}\left\langle \mathbf{\nabla }^{k}f,M_{1}t_{1}\otimes
...\otimes \mathbf{Hess}\Phi M_{j}t_{j}\otimes ...\otimes
M_{k}t_{k}\right\rangle \\
&&+\sum\limits_{\substack{ A\cup B=\{1,...,k\},A\cap B=\emptyset  \\ \#B\leq
k-2}}\left\langle t_{MA}(\partial _{x})\mathbf{\nabla }\Phi ,t_{MB}(\partial
_{x})\mathbf{\nabla }f\right\rangle
\end{eqnarray*}%
\begin{equation*}
t_{MA}(\partial _{x})u:=\left\langle \mathbf{\nabla }^{%
\#A}f,M_{1}t_{i_{j_{1}}}\otimes ...\otimes
M_{\#A}t_{i_{j_{\#A}}}\right\rangle ,\;\;\;\;\;j_{i}\in A.
\end{equation*}%
As in [66], the function 
\begin{equation*}
x\longmapsto \left\langle \mathbf{\nabla }^{k}f(x),M_{1}t_{1}\otimes
...\otimes M_{k}t_{k}\right\rangle
\end{equation*}%
achieves its maximum at $x_{o}.$ Using the notation $\Phi _{x_{i_{A}}}=\Phi
_{x_{i_{\ell _{1}}}...x_{i_{\ell _{r}}}}$ if $A=\left\{ \ell _{1},...\ell
_{r}\right\} \subset \left\{ 1,...k\right\} ,$ we therefore have%
\begin{eqnarray*}
&&\sum_{i_{1}\in \Lambda }g_{x_{i_{1}}...x_{i_{k}}}(x_{o})\rho
(i_{1},...,i_{k})^{2}f_{x_{i_{1}}...x_{i_{k}}}(x_{o}) \\
&\geq &\sum_{s\in \Lambda }\sum_{i_{1}\in \Lambda
}f_{x_{i_{1}}...x_{i_{k}}}(x_{o})f_{x_{s}x_{i_{2}}...x_{i_{k}}}(x_{o})\rho
(i_{1},...,i_{k})^{2}\Phi _{x_{s}x_{i_{1}}}(x_{o}) \\
&&+\sum_{j=2}^{k}\sum_{i_{1}\in \Lambda }\sum_{s\in \Lambda }f_{x_{i_{1}}...%
\underbrace{x_{s}}_{jth}...x_{i_{k}}}(x_{o})f_{x_{i_{1}}...x_{i_{k}}}(x_{o})%
\rho (i_{1},...,i_{k})^{2}\Phi _{x_{s}x_{i_{j}}}(x_{o}) \\
&&+\sum\limits_{\substack{ A\cup B=\{1,...,k\},A\cap B=\emptyset  \\ \#B\leq
k-2  \\ 1\in A}}\left\langle \sum_{i_{1}\in \Lambda }\mathbf{\nabla }\Phi
_{x_{i_{A}}}f_{x_{i_{1}}...x_{i_{k}}}(x_{o})\rho (i_{1},...,i_{k})^{2},%
\mathbf{\nabla }f_{x_{i_{B}}}(x_{o})\right\rangle \\
&&+\sum\limits_{\substack{ A\cup B=\{1,...,k\},A\cap B=\emptyset  \\ \#B\leq
k-2  \\ 1\in B}}\left\langle \mathbf{\nabla }\Phi
_{x_{i_{A}}},\sum_{i_{1}\in \Lambda }\mathbf{\nabla }%
f_{x_{i_{B}}}(x_{o})f_{x_{i_{1}}...x_{i_{k}}}(x_{o})\rho
(i_{1},...,i_{k})^{2}\right\rangle .
\end{eqnarray*}%
Equivalently 
\begin{eqnarray*}
&&\sum_{i_{1}\in \Lambda }g_{x_{i_{1}}...x_{i_{k}}}(x_{o})\rho
(i_{1},...,i_{k})^{2}f_{x_{i_{1}}...x_{i_{k}}}(x_{o}) \\
&\geq &\sum_{s\in \Lambda }\sum_{i_{1}\in \Lambda
}f_{x_{i_{1}}...x_{i_{k}}}(x_{o})f_{x_{s}...x_{i_{k}}}(x_{o})\rho
(i_{1},...,i_{k})^{2}\Phi _{x_{s}x_{i_{1}}}(x_{o}) \\
&&+\sum_{j=2}^{k}\sum_{i_{1}\in \Lambda }\sum_{s\in \Lambda }f_{x_{i_{1}}...%
\underbrace{x_{s}}_{jth}...x_{i_{k}}}(x_{o})f_{x_{i_{1}}...x_{i_{k}}}(x_{o})%
\rho (i_{1},...,i_{k})^{2}\Phi _{x_{s}x_{i_{j}}}(x_{o}) \\
&&+\sum\limits_{\substack{ A\cup B=\{1,...,k\},A\cap B=\emptyset  \\ \#B\leq
k-2  \\ 1\in A}}\sum_{s\in \Lambda }\sum_{i_{1}\in \Lambda }\Phi
_{x_{i_{A}}x_{s}}f_{x_{i_{1}}...x_{i_{k}}}(x_{o})\rho
(i_{1},...,i_{k})^{2}f_{x_{i_{B}}x_{s}}(x_{o}) \\
&&+\sum\limits_{\substack{ A\cup B=\{1,...,k\},A\cap B=\emptyset  \\ \#B\leq
k-2  \\ 1\in B}}\sum_{i_{1}\in \Lambda }\sum_{s\in \Lambda }\Phi
_{x_{i_{A}}x_{s}}f_{x_{i_{B}}x_{s}}(x_{o})f_{x_{i_{1}}...x_{i_{k}}}(x_{o})%
\rho (i_{1},...,i_{k})^{2}.
\end{eqnarray*}%
Now taking summation over $i_{2},...,i_{k,}$ we get%
\begin{eqnarray*}
&&\sum_{i_{2},...,i_{k}\in \Lambda }\sum_{i_{1}\in \Lambda
}g_{x_{i_{1}}...x_{i_{k}}}(x_{o})\rho
(i_{1},...,i_{k})^{2}f_{x_{i_{1}}...x_{i_{k}}}(x_{o}) \\
&\geq &\sum_{i_{2},...,i_{k}\in \Lambda }\sum_{s\in \Lambda }\sum_{i_{1}\in
\Lambda }f_{x_{i_{1}}...x_{i_{k}}}(x_{o})f_{x_{s}...x_{i_{k}}}(x_{o})\rho
(i_{1},...,i_{k})^{2}\Phi _{x_{s}x_{i_{1}}}(x_{o}) \\
&&+\sum_{i_{2},...,i_{k}\in \Lambda }\sum_{j=2}^{k}\sum_{i_{1}\in \Lambda
}\sum_{s\in \Lambda }f_{x_{i_{1}}...\underbrace{x_{s}}%
_{jth}...x_{i_{k}}}(x_{o})f_{x_{i_{1}}...x_{i_{k}}}(x_{o})\rho
(i_{1},...,i_{k})^{2}\Phi _{x_{s}x_{i_{j}}}(x_{o}) \\
&&+\sum_{i_{2},...,i_{k}\in \Lambda }\sum\limits_{\substack{ A\cup
B=\{1,...,k\},A\cap B=\emptyset  \\ \#B\leq k-2  \\ 1\in A}}\sum_{s\in
\Lambda }\sum_{i_{1}\in \Lambda }\Phi
_{x_{i_{A}}x_{s}}f_{x_{i_{1}}...x_{i_{k}}}(x_{o})\rho
(i_{1},...,i_{k})^{2}f_{x_{i_{B}}x_{s}}(x_{o}) \\
&&+\sum_{i_{2},...,i_{k}\in \Lambda }\sum\limits_{\substack{ A\cup
B=\{1,...,k\},A\cap B=\emptyset  \\ \#B\leq k-2  \\ 1\in B}}\sum_{i_{1}\in
\Lambda }\sum_{s\in \Lambda }\Phi
_{x_{i_{A}}x_{s}}f_{x_{i_{B}}x_{s}}(x_{o})f_{x_{i_{1}}...x_{i_{k}}}(x_{o})%
\rho (i_{1},...,i_{k})^{2}.
\end{eqnarray*}%
Next, we propose to estimate each term of the right hand side of this above
inequality. 
\begin{eqnarray*}
&&\sum_{i_{2},...,i_{k}\in \Lambda }\sum_{s\in \Lambda }\sum_{i_{1}\in
\Lambda }f_{x_{i_{1}}...x_{i_{k}}}(x_{o})f_{x_{s}...x_{i_{k}}}(x_{o})\rho
(i_{1},...,i_{k})^{2}\Phi _{x_{s}x_{i_{1}}}(x_{o}) \\
&=&\sum_{i_{2},...,i_{k}\in \Lambda }\left\langle \mathbf{\nabla }%
f_{x_{i_{2}}...x_{i_{k}}}(x_{o}),\mathbf{Hess}\Phi M_{1}t_{1}\right\rangle \\
&=&\sum_{i_{2},...,i_{k}\in \Lambda }\left\langle M_{1}\mathbf{\nabla }%
f_{x_{i_{2}}...x_{i_{k}}}(x_{o}),M_{1}^{-1}\mathbf{Hess}\Phi
M_{1}t_{1}\right\rangle \\
&=&\sum_{i_{2},...,i_{k}\in \Lambda }\left\langle t_{1},M_{1}^{-1}\mathbf{%
Hess}\Phi M_{1}t_{1}\right\rangle \\
&\geq &\delta _{o}\sum_{i_{2},...,i_{k}\in \Lambda }\left\| t_{1}\right\|
^{2} \\
&=&\delta _{o}\sum_{i_{1},...i_{k}\in \Lambda
}f_{x_{i_{1}}...x_{i_{k}}}^{2}(x_{o})\rho (i_{1},...,i_{k})^{2}.
\end{eqnarray*}%
Similarly, it is easy to see that 
\begin{eqnarray*}
&&\sum_{i_{2},...,i_{k}\in \Lambda }\sum_{j=2}^{k}\sum_{i_{1}\in \Lambda
}\sum_{s\in \Lambda }f_{x_{i_{1}}...\underbrace{x_{s}}%
_{jth}...x_{i_{k}}}(x_{o})f_{x_{i_{1}}...x_{i_{k}}}(x_{o})\rho
(i_{1},...,i_{k})^{2}\Phi _{x_{s}x_{i_{j}}}(x_{o}) \\
&\geq &(k-1)\delta _{0}\sum_{i_{1},...i_{k}\in \Lambda
}f_{x_{i_{1}}...x_{i_{k}}}^{2}(x_{o})\rho (i_{1},...,i_{k})^{2}
\end{eqnarray*}%
To estimate the last two terms, we have%
\begin{eqnarray*}
&&\sum_{i_{2},...,i_{k}\in \Lambda }\sum\limits_{\substack{ A\cup
B=\{1,...,k\},A\cap B=\emptyset  \\ \#B\leq k-2  \\ 1\in A}}\sum_{i_{1}\in
\Lambda }\sum_{s\in \Lambda }\left| \Phi
_{x_{i_{A}}x_{s}}f_{x_{i_{1}}...x_{i_{k}}}(x_{o})\rho
(i_{1},...,i_{k})^{2}f_{x_{i_{B}}x_{s}}(x_{o})\right| \\
&\leq &\left[ \sum_{i_{1},...,i_{k}\in \Lambda
}f_{x_{i_{1}}...x_{i_{k}}}^{2}(x_{o})\rho (i_{1},...,i_{k})^{2}\right]
^{1\backslash 2}\times \\
&&\left[ \sum_{i_{1},...,i_{k}\in \Lambda }\left( \sum\limits_{\substack{ %
A\cup B=\{1,...,k\},A\cap B=\emptyset  \\ \#B\leq k-2  \\ 1\in A}}\sum_{s\in
\Lambda }\left| \Phi _{x_{i_{A}}x_{s}}\rho
(i_{1},...,i_{k})f_{x_{i_{B}}x_{s}}(x_{o})\right| \right) ^{2}\right]
^{1\backslash 2}.
\end{eqnarray*}%
To estimate the second factor of the right hand side of this last above
inequality, we have%
\begin{eqnarray*}
&&\sum_{i_{1},...,i_{k}\in \Lambda }\left( \sum\limits_{\substack{ A\cup
B=\{1,...,k\},A\cap B=\emptyset  \\ \#B\leq k-2  \\ 1\in A}}\sum_{s\in
\Lambda }\left| \Phi _{x_{i_{A}}x_{s}}(x_{o})\rho
(i_{1},...,i_{k})f_{x_{i_{B}}x_{s}}(x_{o})\right| \right) ^{2} \\
&\leq &C_{k}\sum_{i_{1},...,i_{k}\in \Lambda }\sum\limits_{\substack{ A\cup
B=\{1,...,k\},A\cap B=\emptyset  \\ \#B\leq k-2  \\ 1\in A}}\left(
\sum_{s\in \Lambda }\Phi _{x_{i_{A}}x_{s}}(x_{o})\rho
(i_{1},...,i_{k})f_{x_{i_{B}}x_{s}}(x_{o})\right) ^{2} \\
&\leq &C_{k}\sum_{i_{1},...,i_{k}\in \Lambda }\sum\limits_{\substack{ A\cup
B=\{1,...,k\},A\cap B=\emptyset  \\ \#B\leq k-2  \\ 1\in A}}\left(
\sum_{s\in \Lambda }\Phi _{x_{i_{A}}x_{s}}^{2}(x_{o})\rho
^{2}(i_{1},...,i_{k})\right) \times \\
&&\left( \sum_{s\in \Lambda }\rho
^{2}(i_{1},...,i_{k})f_{x_{i_{B}}x_{s}}^{2}(x_{o})\right) \\
&\leq &C_{k}\sum_{i_{1},...,i_{k}\in \Lambda }\sum\limits_{\substack{ A\cup
B=\{1,...,k\},A\cap B=\emptyset  \\ \#B\leq k-2  \\ 1\in A}}\left(
\sum_{s\in \Lambda }\Phi _{x_{i_{A}}x_{s}}^{2}(x_{o})e^{2\kappa d\left(
\left\{ i_{j}:j\in A\right\} ,S_{g}\right) }\right) \times \\
&&\left( \sum_{s\in \Lambda }e^{2\kappa d\left( \left\{ i_{j}:j\in B\right\}
\cup \left\{ s\right\} ,S_{g}\right) }f_{x_{i_{B}}x_{s}}^{2}(x_{o})\right) \\
&\leq &C_{k}.
\end{eqnarray*}%
This last inequality \ above follows from the induction assumption and that
of $\Phi .$

Thus,%
\begin{eqnarray*}
&&\sum_{i_{2},...,i_{k}\in \Lambda }\sum\limits_{\substack{ A\cup
B=\{1,...,k\},A\cap B=\emptyset  \\ \#B\leq k-2  \\ 1\in A}}\sum_{i_{1}\in
\Lambda }\sum_{s\in \Lambda }\Phi
_{x_{i_{A}}x_{s}}f_{x_{i_{1}}...x_{i_{k}}}(x_{o})\rho
(i_{1},...,i_{k})^{2}f_{x_{i_{B}}x_{s}}(x_{o}) \\
&\geq &-C_{k}\left[ \sum_{i_{1},...,i_{k}\in \Lambda
}f_{x_{i_{1}}...x_{i_{k}}}^{2}(x_{o})\rho (i_{1},...,i_{k})^{2}\right]
^{1\backslash 2}.
\end{eqnarray*}%
Similarly, we have 
\begin{eqnarray*}
&&\sum_{i_{2},...,i_{k}\in \Lambda }\sum\limits_{\substack{ A\cup
B=\{1,...,k\},A\cap B=\emptyset  \\ \#B\leq k-2  \\ 1\in B}}\sum_{i_{1}\in
\Lambda }\sum_{s\in \Lambda }\Phi
_{x_{i_{A}}x_{s}}f_{x_{i_{B}}x_{s}}(x_{o})f_{x_{i_{1}}...x_{i_{k}}}(x_{o})%
\rho (i_{1},...,i_{k})^{2} \\
&\geq &-C_{k}\left[ \sum_{i_{1},...,i_{k}\in \Lambda
}f_{x_{i_{1}}...x_{i_{k}}}^{2}(x_{o})\rho (i_{1},...,i_{k})^{2}\right]
^{1\backslash 2}.
\end{eqnarray*}%
We then finally get 
\begin{eqnarray*}
&&\sum_{i_{1},...,i_{k}\in \Lambda }g_{x_{i_{1}}...x_{i_{k}}}(x_{o})\rho
(i_{1},...,i_{k})^{2}f_{x_{i_{1}}...x_{i_{k}}}(x_{o}) \\
&\geq &k\delta _{o}\sum_{i_{1},...i_{k}\in \Lambda
}f_{x_{i_{1}}...x_{i_{k}}}^{2}(x_{o})\rho (i_{1},...,i_{k})^{2} \\
&&-C_{k}\left[ \sum_{i_{1},...,i_{k}\in \Lambda
}f_{x_{i_{1}}...x_{i_{k}}}^{2}(x_{o})\rho (i_{1},...,i_{k})^{2}\right]
^{1\backslash 2}.
\end{eqnarray*}%
If 
\begin{equation*}
\sum_{i_{1},...,i_{k}\in \Lambda }f_{x_{i_{1}}...x_{i_{k}}}^{2}(x_{o})\rho
(i_{1},...,i_{k})^{2}=0
\end{equation*}%
then there is nothing to prove, otherwise we have, after using
Cauchy-Schwartz and dividing by 
\begin{equation*}
\sum_{i_{1},...,i_{k}\in \Lambda }f_{x_{i_{1}}...x_{i_{k}}}^{2}(x_{o})\rho
(i_{1},...,i_{k})^{2},
\end{equation*}%
\begin{eqnarray*}
&&\left( \sum_{i_{1},...i_{k}\in \Lambda
}f_{x_{i_{1}}...x_{i_{k}}}^{2}(x_{o})\rho (i_{1},...,i_{k})^{2}\right) ^{1/2}
\\
&\leq &\frac{1}{k\delta _{o}}\left( \sum_{i_{1},...,i_{k}\in \Lambda
}g_{x_{i_{1}}...x_{i_{k}}}^{2}(x_{o})\right) ^{1/2}+C_{k,g} \\
&\leq
&C_{k,g}.\;\;\;\;\;\;\;\;\;\;\;\;\;\;\;\;\;\;\;\;\;\;\;\;\;\;\;\;\;\;\;\;\;%
\;\blacksquare
\end{eqnarray*}

\section{Relaxing the Assumptions of Compact support}

As in [8], we consider the family cutoff functions 
\begin{equation}
\chi =\chi _{\varepsilon }
\end{equation}%
$(\varepsilon \in \lbrack 0,1])$ in $\mathcal{C}_{o}^{\infty }(\mathbb{R})$
with value in $[0,1]$ such that%
\begin{equation*}
\left\{ 
\begin{array}{c}
\chi =1\text{ \ \ \ \ \ \ \ \ \ \ \ \ \ \ \ \ for }\left| t\right| \leq
\varepsilon ^{-1} \\ 
\left| \chi ^{(k)}(t)\right| \leq C_{k}\dfrac{\varepsilon }{\left| t\right|
^{k}}\text{ \ \ \ \ \ \ \ \ \ \ \ \ \ \ \ \ for }k\in \mathbb{N\ }.\text{\ \
\ \ \ \ \ \ \ \ \ \ \ \ \ \ \ \ \ \ \ \ \ }%
\end{array}%
\right.
\end{equation*}%
We then introduce%
\begin{equation}
\Psi _{\varepsilon }(x)=\chi _{\varepsilon }(\left| x\right| )\Psi (x)\text{
\ \ \ \ }x\in \mathbb{R}^{\Lambda }
\end{equation}%
and 
\begin{equation}
g_{\varepsilon }(x)=\chi _{\varepsilon }(\left| x\right| )g(x)\text{ \ \ \ \
\ \ }x\in \mathbb{R}^{\Gamma }.
\end{equation}%
A straightforward computation (see [66]) shows that $\Psi _{\varepsilon }(x)$
and $g_{\varepsilon }(x).$satisfy 
\begin{equation}
\left| \partial ^{\alpha }\mathbf{\nabla }\Psi _{\varepsilon }\right| \leq
C_{\alpha }+\mathcal{O}_{\alpha ,\Lambda }(\varepsilon ),\;\ \ \ \ \ \forall
\alpha \in \mathbb{N}^{\left| \Lambda \right| }.
\end{equation}%
and%
\begin{equation}
\left| \partial ^{\alpha }\mathbf{\nabla }g_{\varepsilon }\right| \leq
C_{\alpha }+\mathcal{O}_{\alpha ,\Lambda }(\varepsilon ),\;\ \ \ \ \ \forall
\alpha \in \mathbb{N}^{\left| \Gamma \right| },
\end{equation}%
and that%
\begin{equation}
M^{-1}\mathbf{Hess}\Phi _{\varepsilon }(x)M\geq \delta ^{\prime
},\;\;\;\;\;\;\;\;0<\delta ^{\prime }<1.\text{ \ \ in the sense of (17)}
\end{equation}%
It then only remains to check that 
\begin{equation}
\sum_{j,i_{1},...,i_{k}\in \Lambda }\Psi _{\varepsilon
_{x_{j}x_{i_{1}}...x_{i_{k}}}}^{2}(x)e^{2\kappa d(\left\{
i_{1},...,i_{k}\right\} ,S_{g})}\leq C_{k}+\text{ }\mathcal{O}_{k,\Lambda
}(\varepsilon )\;\;\;\;\;\;\;\forall x\in \mathbb{R}^{\Lambda },\forall
k\geq 2
\end{equation}%
where $C_{k}$ is a positive constant that does not depend on $f$ and $%
\Lambda $.%
\begin{equation*}
\Psi _{\varepsilon }(x)=\chi _{\varepsilon }(r)\Psi (x)
\end{equation*}%
Let $\alpha $ be such that $\left| \alpha \right| \geq 3.$ Using Leibniz's
formula, we have%
\begin{eqnarray}
\left| \partial ^{\alpha }\Psi _{\varepsilon }\right| &\leq &\sum_{\beta
\leq \alpha }\dbinom{\alpha }{\beta }\left| \partial ^{\beta }\chi
_{\varepsilon }(r)\partial ^{\alpha -\beta }\Psi \right| \\
&\leq &\left| \partial ^{\alpha }\chi _{\varepsilon }(r)\Psi \right| +\left|
\partial ^{\alpha }\Psi \right| +\sum_{\substack{ \beta <\alpha  \\ \beta
\neq 0}}\dbinom{\alpha }{\beta }\left| \partial ^{\beta }\chi _{\varepsilon
}(r)\partial ^{\alpha -\beta }\Psi \right| .
\end{eqnarray}%
Assuming that $\Psi (0)=0$ and write 
\begin{equation*}
\Psi (x)=\int_{0}^{1}x\cdot \mathbf{\nabla }\Psi (sx)ds
\end{equation*}%
\begin{eqnarray*}
\left| \partial ^{\alpha }\chi _{\varepsilon }(r)\Psi (x)\right| &\leq
&\sum_{j_{1}\in \Lambda }\int_{0}^{1}\left| x_{j_{1}}\partial ^{\alpha }\chi
_{\varepsilon }(r)\Psi _{x_{j_{1}}}(sx)\right| ds \\
&\leq &C\left| r\partial ^{\alpha }\chi _{\varepsilon }(r)\right| .
\end{eqnarray*}%
Now using the fact that 
\begin{equation*}
r\partial ^{\alpha }\chi _{\varepsilon }(r)=\mathcal{O}_{\alpha
}(\varepsilon ),
\end{equation*}%
we have 
\begin{equation*}
\left| \partial ^{\alpha }\chi _{\varepsilon }(r)\Psi (x)\right| =\mathcal{O}%
_{\alpha ,\Lambda }(\varepsilon ).
\end{equation*}%
Finally, using the fact that 
\begin{equation}
\partial ^{\beta }\chi _{\varepsilon }(r)=\mathcal{O}_{\beta }(\varepsilon
)\;\ \ \text{for every }\left| \beta \right| \geq 1\text{,}
\end{equation}%
it is then easy to see that 
\begin{equation}
\left( \sum_{\substack{ \beta <\alpha  \\ \beta \neq 0}}\dbinom{\alpha }{%
\beta }\left| \partial ^{\beta }\chi _{\varepsilon }(r)\partial ^{\alpha
-\beta }\Psi \right| \right) ^{2}=\mathcal{O}_{\alpha ,\Lambda }(\varepsilon
).
\end{equation}%
Thus%
\begin{equation}
\sum_{j,i_{1},...,i_{k}\in \Lambda }\Psi _{\varepsilon
_{x_{j}x_{i_{1}}...x_{i_{k}}}}^{2}(x)e^{2\kappa d(\left\{
i_{1},...,i_{k}\right\} ,S_{g})}\leq C_{k,g}+\text{ }\mathcal{O}_{k,\Lambda
}(\varepsilon )\;\;\;\;\;\;\;\forall x\in \mathbb{R}^{\Lambda },\forall
k\geq 2.
\end{equation}

Now using the arguments developed in [8] (see also [66]) about the
convergence of the corresponding solutions as $\varepsilon \rightarrow 0,$
we obtain:

\begin{proposition}
\textit{If }$g(0)=\Psi (0)=0,$ \textit{then Proposition 2 holds without the
assumptions of compact support on }$\Psi $ \textit{and }$g.$
\end{proposition}

\section{The n-Point Correlation Functions}

The higher order correlation is defined as 
\begin{equation}
\left\langle g_{1},...,g_{k}\right\rangle :=\left\langle \left(
g_{1}-\left\langle g_{1}\right\rangle \right) ...\left( g_{k}-\left\langle
g_{k}\right\rangle \right) \right\rangle .
\end{equation}%
For simplicity we shall take $k=3$ and $\Phi $ is as in proposition 2.

Let $g_{1},g_{2},$ and $g_{3}$ be smooth functions satisfying $\left(
15\right) $ and $f_{i}$ $i=1,2,3$ shall denote the unique solution of the
system%
\begin{equation}
\left\{ 
\begin{tabular}{l}
$-\mathbf{\Delta }f_{i}+\mathbf{\nabla }\Phi \cdot \mathbf{\nabla }%
f_{i}=g_{i}-\left\langle g_{i}\right\rangle _{_{L^{2}(\mu )}}$ \\ 
$\left\langle f_{i}\right\rangle _{L^{2}(\mu )}=0.$%
\end{tabular}%
\right.
\end{equation}%
Recall that 
\begin{equation*}
\mathbf{\nabla }f_{i}=A_{\Phi }^{\left( 1\right) ^{-1}}\mathbf{\nabla }g_{i}.
\end{equation*}%
For an arbitrary smooth function $c,$ it is easy to see that 
\begin{equation*}
\left\langle c(x)\left( g_{i}-\left\langle g_{i}\right\rangle \right)
\right\rangle =\left\langle \mathbf{\nabla }f_{i}\cdot \mathbf{\nabla }%
c\right\rangle .
\end{equation*}%
A direct computation shows that 
\begin{eqnarray*}
\left\langle g_{1},g_{2,}g_{3}\right\rangle &=&\left\langle \mathbf{\nabla }%
f_{3}\cdot \left( \mathbf{Hess}f_{1}\right) \mathbf{\nabla }%
g_{2}\right\rangle +\left\langle \mathbf{\nabla }f_{3}\cdot \left( \mathbf{%
Hess}g_{2}\right) \mathbf{\nabla }f_{1}\right\rangle \\
&&+\left\langle \mathbf{\nabla }f_{2}\cdot \left( \mathbf{Hess}f_{1}\right) 
\mathbf{\nabla }g_{3}\right\rangle +\left\langle \mathbf{\nabla }f_{2}\cdot
\left( \mathbf{Hess}g_{3}\right) \mathbf{\nabla }f_{1}\right\rangle .
\end{eqnarray*}%
Let us now estimate each term of the right and side of this equality.

Using Cauchy-Schwartz, and proposition 2, it is easy to see that 
\begin{equation*}
\left| \left\langle \mathbf{\nabla }f_{3}\cdot \left( \mathbf{Hess}%
f_{1}\right) \mathbf{\nabla }g_{2}\right\rangle \right| \leq Ce^{-\kappa
_{1}d\left( S_{g_{2}},S_{g_{1}}\right) }
\end{equation*}%
\begin{equation*}
\left| \left\langle \mathbf{\nabla }f_{3}\cdot \left( \mathbf{Hess}%
g_{2}\right) \mathbf{\nabla }f_{1}\right\rangle \right| \leq Ce^{-\kappa
_{1}d\left( S_{g_{2}},S_{g_{1}}\right) },
\end{equation*}%
\begin{equation*}
\left| \left\langle \mathbf{\nabla }f_{2}\cdot \left( \mathbf{Hess}%
f_{1}\right) \mathbf{\nabla }g_{3}\right\rangle \right| \leq Ce^{-\kappa
_{1}d\left( S_{g_{3}},S_{g_{1}}\right) }
\end{equation*}%
and 
\begin{equation*}
\left| \left\langle \mathbf{\nabla }f_{2}\cdot \left( \mathbf{Hess}%
g_{3}\right) \mathbf{\nabla }f_{1}\right\rangle \right| \leq Ce^{-\kappa
_{1}d\left( S_{g_{3}},S_{g_{1}}\right) }
\end{equation*}%
Here the constants $C$ only depends on the size of the support of the $g_{i}$%
's$.$ and $\kappa _{1}>0.$

Thus 
\begin{equation*}
\left| \left\langle g_{1},g_{2,}g_{3}\right\rangle \right| \leq C\left[
e^{-\kappa _{1}d\left( S_{g_{2}},S_{g_{1}}\right) }+e^{-\kappa _{1}d\left(
S_{g_{3}},S_{g_{1}}\right) }\right]
\end{equation*}%
If $g_{1}=x_{i}$, $g_{2}=x_{j},$ and $g_{3}=x_{k},$ we obtain%
\begin{equation*}
\left| \left\langle \left( x_{i}-\left\langle x_{i}\right\rangle \right)
\left( x_{j}-\left\langle x_{j}\right\rangle \right) \left(
x_{k}-\left\langle x_{k}\right\rangle \right) \right\rangle \right| \leq C 
\left[ e^{-\kappa _{1}d\left( i,j\right) }+e^{-\kappa _{1}d\left( i,k\right)
}\right] .
\end{equation*}%
Thus if $d>1,$ we obtain this weak exponential decay of the truncated
correlations in the sense that the exponential decay occurs as you
simultaneously pull the spins away from a fixed one. Note that in the one
dimensional case, we obtain a stronger exponential decay due to the fact
that 
\begin{equation*}
i\leq j\leq k\Longrightarrow d(i,k)=d(i,j)+d(j,k).
\end{equation*}%
This was already pointed out in [8].

\section{The Analyticity of the Pressure}

In this section, we attempt to study a direct method for the analyticity of
the pressure for certain classical convex unbounded spin systems. It is
central in Statistical Mechanics to study the differentiability or even the
analyticity of the pressure with respect to some distinguished thermodynamic
parameters such as temperature, chemical potential or external field. In
fact the analytic behavior of the pressure is the classical thermodynamic
indicator for the absence or existence of phase transition. The most famous
result on the analyticity of the pressure is the circle theorem of Lee and
Yang [28]. This theorem asserts the following: consider a $\left\{
-1,1\right\} -$valued spin system with ferromagnetic pair interaction and
external field $h$ and regard the quantity $z=e^{h}$ as a complex parameter,
then all zeroes of all partition functions (with free boundary condition),
considered as functions of $z$ lie in the complex unit circle. This theorem
readily implies that the pressure is an analytic function of $h$ in the
region $h>0$ and $h<0.$ Heilmann [29] showed that the assumption of pair
interaction is necessary. A transparent approach to the circle theorem was
found by Asano [30] and developed further by Ruelle [31],[32], Slawny [33],
and Gruber et al [34]. Griffiths [35] and Griffiths-Simon [36] found a
method of extending the Lee-Yang theorem to real-valued spin systems with a
particular type of a priory measure. Newman [37] proved the Lee-Yang theorem
for every a priory measure which satisfies this theorem in the particular
case of no interaction. Dunlop [38],[39] studied the zeroes of the partition
functions for the plane rotor model. A general Lee-Yang theorem for
multicomponent systems was finally proved by Lieb and Sokal [40]. For
further references see Glimm and Jaffe [41].

The Lee-Yang theorem and its variants depend on the ferromagnetic character
of the interaction. There are various other way of proving the infinite
differentiability or the analyticity of the pressure for (ferromagnetic and
non ferromagnetic) systems at high temperatures, or at low temperatures, or
at large external fields. Most of these take advantage of a sufficiently
rapid decay of correlations and /or cluster expansion methods. Here is a
small sample of relevant references. Bricmont, Lebowitz and Pfister [42],
Dobroshin [43], Dobroshin and Sholsman [44],[45], Duneau et al
[46],[47],[48], Glimm and Jaffe [41],[49], Israel [50], Kotecky and Preiss
[51], Kunz [52], Lebowitz [53],[54], Malyshev [55], Malychev and Milnos [56]
and Prakash [57]. M. Kac and J.M. Luttinger [58] obtained a formula for the
pressure in terms of irreducible distribution functions.

We propose a new way of analyzing the analyticity of the pressure for
certain unbounded models through a representation by means of the Witten
Laplacians of the coefficients in the Taylor series expansion. The methods
known up to now rely on complicated indirect arguments.

\section{Towards the analyticity of the Pressure}

Let $\Lambda $\ be a finite domain in $\mathbb{Z}^{d}\;(d\geq 1)$ and
consider the Hamiltonian of the phase space given by,%
\begin{equation}
\Phi (x)=\Phi _{\Lambda }(x)=\frac{x^{2}}{2}+\Psi (x),\;\ \ \ \ \;x\in 
\mathbb{R}^{\Lambda }.
\end{equation}%
where%
\begin{equation}
\left| \partial ^{\alpha }\mathbf{\nabla }\Psi \right| \leq C_{\alpha },\;\
\ \ \ \ \forall \alpha \in \mathbb{N}^{\left| \Lambda \right| },
\end{equation}%
\begin{equation}
\mathbf{Hess}\Phi (x)\geq \delta _{o},\;\;\;\;\;\;\;\;\;\;0<\delta _{o}<1.
\end{equation}%
Let $g$ is a smooth function on $\mathbb{R}^{\Gamma }$ with lattice support $%
S_{g}=\Gamma .$ We identified with $\tilde{g}$ defined on $\mathbb{R}%
^{\Lambda }$ by 
\begin{equation}
\tilde{g}(x)=g(x_{\Gamma })\text{ \ \ where }x=\left( x_{i}\right) _{i\in
\Lambda }\text{ \ and }x_{\Gamma }=\left( x_{i}\right) _{i\in \Gamma }
\end{equation}%
and satisfying%
\begin{equation}
\left| \partial ^{\alpha }\mathbf{\nabla }g\right| \leq C_{\alpha
}\;\;\;\;\;\;\;\;\forall \alpha \in \mathbb{N}^{\left| \Gamma \right| }
\end{equation}%
Under the additional assumptions that $\Psi $ is compactly supported in $%
\mathbb{R}^{\Lambda }$ and $g$ is compactly supported in $\mathbb{R}^{\Gamma
},$ it was proved in [66] (see also [8]) that the equation 
\begin{equation*}
\left\{ 
\begin{tabular}{l}
$-\mathbf{\Delta }f+\mathbf{\nabla }\Phi \cdot \mathbf{\nabla }%
f=g-\left\langle g\right\rangle $ \\ 
$\left\langle f\right\rangle _{L^{2}(\mu )}=0$%
\end{tabular}%
\right.
\end{equation*}%
has a unique smooth solution satisfying $\mathbf{\nabla }^{k}f(x)\rightarrow
0$ as $\left| x\right| \rightarrow \infty $ for every $k\geq 1$.

Recall also that $\mathbf{\nabla }f$ is a solution of the system%
\begin{equation}
\left( -\mathbf{\Delta +\nabla }\Phi \cdot \mathbf{\nabla }\right) \mathbf{%
\nabla }f+\mathbf{Hess}\Phi \mathbf{\nabla }f=\mathbf{\nabla }g\;\;\;\;\text{%
in }\;\mathbb{R}^{\Lambda }.
\end{equation}%
As in [66] and [8], these assumptions will be relaxed later on.

Let 
\begin{equation}
\Phi _{\Lambda }^{t}(x)=\Phi (x)-tg(x),
\end{equation}%
where $x=(x_{i})_{i\in \Lambda }$, and assume additionally that $g$
satisfies 
\begin{equation}
\mathbf{Hess}g\leq C.
\end{equation}%
We consider the following perturbation \ 
\begin{equation}
\theta _{\Lambda }(t)=\log \left[ \int dxe^{-\Phi _{\Lambda }^{t}(x)}\right]
.
\end{equation}%
Denote by 
\begin{equation}
Z_{t}=\int dxe^{-\Phi _{\Lambda }^{t}(x)}
\end{equation}%
and 
\begin{equation}
<\cdot >_{t,\Lambda }=\frac{\int \cdot \;dxe^{-\Phi _{\Lambda }^{t}(x)}}{%
Z_{t}}.
\end{equation}

\section{Parameter Dependency of the Solution}

From the assumptions made on $\Phi $ and $g,$ it is easy to see that there
exists $T>0$ such that or every $t\in \lbrack 0,T)$, $\Phi _{\Lambda
}^{t}(x) $ satisfies all the assumptions required for the solvability,
regularity and asymptotic behavior of the solution $f(t)$ associated with
the potential $\Phi _{\Lambda }^{t}(x).$ Thus, each $t\in \lbrack 0,T)$ is
associated with a unique $C^{\infty }-$solution, $f(t)$ of the equation 
\begin{equation*}
\left\{ 
\begin{tabular}{l}
$A_{\Phi _{\Lambda }^{t}}^{(0)}f(t)=g-\left\langle g\right\rangle
_{_{L^{2}(\mu )}}$ \\ 
$\left\langle f(t)\right\rangle _{L^{2}(\mu )}=0.$%
\end{tabular}%
\right.
\end{equation*}%
Hence,%
\begin{equation}
A_{\Phi _{\Lambda }^{t}}^{(1)}\mathbf{v}(t)=\mathbf{\nabla }g
\end{equation}%
where $\mathbf{v}(t)=\mathbf{\nabla }f(t).$ Notice that the map 
\begin{equation*}
t\longmapsto \mathbf{v}(t)
\end{equation*}%
is well defined and 
\begin{equation*}
\left\{ \mathbf{v}(t):t\in \lbrack 0,T)\right\}
\end{equation*}%
is a family of smooth solutions on $\mathbb{R}^{\Lambda }$ satisfying 
\begin{equation*}
\partial ^{\alpha }\mathbf{v}(t)\rightarrow 0\;\;\;\text{as }\left| x\right|
\rightarrow \infty \ \ \ \ \ \ \ \forall \alpha \in \mathbb{N}^{\left|
\Lambda \right| }\ \text{\ and for each }t\in \lbrack 0,T)
\end{equation*}%
and corresponding to the family of potential 
\begin{equation}
\left\{ \Phi _{\Lambda }^{t}:t\in \lbrack 0,T)\right\} .
\end{equation}%
Let us now verify that $\mathbf{v}$ is a smooth function of $t\in (0,T).$We
need to prove that for each $t\in (0,T),$ the limit%
\begin{equation*}
\lim\limits_{\varepsilon \rightarrow 0}\frac{\mathbf{v}(t+\varepsilon )-%
\mathbf{v}(t)}{\varepsilon }
\end{equation*}%
exists. Let 
\begin{equation*}
\mathbf{v}^{\varepsilon }(t)=\frac{\mathbf{v}(t+\varepsilon )-\mathbf{v}(t)}{%
\varepsilon }.
\end{equation*}%
We use a technique based on regularity estimates to get a uniform control of 
$\mathbf{v}^{\varepsilon }(t)$ with respect to $\varepsilon .$

With $\varepsilon $ small enough, we have%
\begin{eqnarray*}
0 &=&-\mathbf{\Delta }\left[ \dfrac{\mathbf{v}(t+\varepsilon )-\mathbf{v}(t)%
}{\varepsilon }\right] +\dfrac{\mathbf{\nabla }\Phi ^{t+\varepsilon }\cdot 
\mathbf{\nabla v}(t+\varepsilon )-\mathbf{\nabla }\Phi ^{t}\cdot \mathbf{%
\nabla v}(t)}{\varepsilon } \\
&&+\dfrac{\mathbf{Hess}\Phi ^{t+\varepsilon }\mathbf{v}(t+\varepsilon )-%
\mathbf{Hess}\Phi ^{t}\mathbf{v}(t)}{\varepsilon }.
\end{eqnarray*}%
Equivalently,%
\begin{eqnarray*}
&&-\mathbf{\Delta }\left[ \dfrac{\mathbf{v}(t+\varepsilon )-\mathbf{v}(t)}{%
\varepsilon }\right] +\dfrac{\mathbf{\nabla }\Phi ^{t+\varepsilon }\cdot 
\mathbf{\nabla }\left[ \mathbf{v}(t+\varepsilon )-\mathbf{v}(t)\right] }{%
\varepsilon } \\
&&+\mathbf{Hess}\Phi ^{t+\varepsilon }\left( \dfrac{\mathbf{v}(t+\varepsilon
)-\mathbf{v}(t)}{\varepsilon }\right) \\
&=&-\left( \dfrac{\mathbf{Hess}\Phi ^{t+\varepsilon }-\mathbf{Hess}\Phi ^{t}%
}{\varepsilon }\right) \mathbf{v}(t)-\left( \dfrac{\mathbf{\nabla }\Phi
^{t+\varepsilon }-\mathbf{\nabla }\Phi ^{t}}{\varepsilon }\right) \cdot 
\mathbf{\nabla v}(t) \\
&=&\mathbf{Hess}g\mathbf{v}(t)+\mathbf{\nabla }g\cdot \mathbf{\nabla v}(t)
\end{eqnarray*}%
and%
\begin{eqnarray*}
&&-\mathbf{\Delta v}^{\varepsilon }(t)+\mathbf{\nabla }\Phi ^{t+\varepsilon
}\cdot \mathbf{\nabla v}^{\varepsilon }(t)+\mathbf{Hess}\Phi ^{t+\varepsilon
}\mathbf{v}^{\varepsilon }(t) \\
&=&\mathbf{Hess}g\mathbf{v}(t)+\mathbf{\nabla }g\cdot \mathbf{\nabla v}(t)
\end{eqnarray*}%
Let $\mathbf{w}(t)$ be the unique $C^{\infty }-$solution of the system 
\begin{equation}
-\mathbf{\Delta w}(t)+\mathbf{\nabla }\Phi ^{t}\cdot \mathbf{\nabla w}(t)+%
\mathbf{Hess}\Phi ^{t}\mathbf{w}(t)=\mathbf{Hess}g\mathbf{v}(t)+\mathbf{%
\nabla }g\cdot \mathbf{\nabla v}(t).
\end{equation}%
Recall that the unitary transformation $U_{\Phi ^{t+\varepsilon },}$ allows
us to reduce%
\begin{equation}
\begin{array}{c}
-\mathbf{\Delta v}^{\varepsilon }(t)+\mathbf{\nabla }\Phi ^{t+\varepsilon
}\cdot \mathbf{\nabla v}^{\varepsilon }(t)+\mathbf{Hess}\Phi ^{t+\varepsilon
}\mathbf{v}^{\varepsilon }(t) \\ 
=\mathbf{Hess}g\mathbf{v}(t)+\mathbf{\nabla }g\cdot \mathbf{\nabla v}(t)%
\end{array}%
\end{equation}%
into%
\begin{equation}
\left. 
\begin{array}{c}
\left( \mathbf{-\Delta +}\dfrac{\left| \mathbf{\nabla }\Phi ^{t+\varepsilon
}\right| ^{2}}{4}-\dfrac{\mathbf{\Delta }\Phi ^{t+\varepsilon }}{2}\right) 
\mathbf{V}^{\varepsilon }+\mathbf{Hess}\Phi ^{t+\varepsilon }\mathbf{V}%
^{\varepsilon }= \\ 
\left[ \mathbf{Hess}g\mathbf{v}(t)+\mathbf{\nabla }g\cdot \mathbf{\nabla v}%
(t)\right] e^{-\Phi ^{t+\varepsilon }/2}%
\end{array}%
\right.
\end{equation}%
where $\mathbf{V}^{\varepsilon }=\mathbf{v}^{\varepsilon }(t)e^{-\Phi
^{t+\varepsilon }/2}.$

\begin{remark}
This unitary transformation already mentioned in the introduction was
introduced in the proof of the existence of solution (see [66] ) to avoid
working with the weighted spaces $L^{2}(\mathbb{R}^{\Lambda },\mathbb{R}%
^{\Lambda },e^{-\Phi }dx).$ The proof was based on Hilbert space method. The
method consists of determining an appropriate function space and an operator
which is a natural realization of the problem. In this particular problem,
the function spaces to be considered are the Sobolev spaces $B_{\Phi }^{k}(%
\mathbb{R}^{\Lambda })$ defined by 
\begin{equation*}
B_{\Phi }^{k}(\mathbb{R}^{\Lambda })=\left\{ u\in L^{2}(\mathbb{R}^{\Lambda
}):Z_{\Phi }^{\ell }\partial ^{\alpha }u\in L^{2}(\mathbb{R}^{\Lambda
})\;\forall \;\ell +\left| \alpha \right| \leq k\right\} .
\end{equation*}%
where 
\begin{equation}
Z_{\Phi }=\frac{\left| \mathbf{\nabla }\Phi \right| }{2}
\end{equation}%
These are subspaces of the well known Sobolev spaces $W^{k,2}(\mathbb{R}%
^{\Lambda }),$ $k$ $\in \mathbb{N}$.
\end{remark}

Taking scalar product with $\mathbf{V}^{\varepsilon }$ on both sides of
(51), we get%
\begin{equation}
\left. 
\begin{array}{c}
\left\| \left( \mathbf{\nabla +}\dfrac{\mathbf{\nabla }\Phi ^{t+\varepsilon }%
}{2}\right) \mathbf{V}^{\varepsilon }\right\| _{L^{2}}^{2}+\int \mathbf{Hess}%
\Phi ^{t+\varepsilon }\mathbf{V}^{\varepsilon }\cdot \mathbf{V}^{\varepsilon
}dx= \\ 
\int \left[ \mathbf{Hess}g\mathbf{v}(t)+\mathbf{\nabla }g\cdot \mathbf{%
\nabla v}(t)\right] e^{-\Phi ^{t+\varepsilon }/2}\cdot \mathbf{V}%
^{\varepsilon }dx.%
\end{array}%
\right.
\end{equation}%
Now using the uniform strict convexity on the left hand side and
Cauchy-Schwartz on the right hand side, we obtain 
\begin{equation}
\left\| \mathbf{V}^{\varepsilon }\right\| _{B^{0}}\leq C_{t}\text{ \ \ \ \
for small enough }\varepsilon .
\end{equation}%
We then deduce that 
\begin{equation}
\left( \mathbf{-\Delta +}\frac{\left| \mathbf{\nabla }\Phi ^{t+\varepsilon
}\right| ^{2}}{4}\right) \mathbf{V}^{\varepsilon }=\tilde{q}_{\varepsilon }
\end{equation}%
where%
\begin{equation}
\tilde{q}_{\varepsilon }=\left[ \mathbf{Hess}g\mathbf{v}(t)+\mathbf{\nabla }%
g\cdot \mathbf{\nabla v}(t)\right] e^{-\Phi ^{t+\varepsilon }/2}\mathbf{+}%
\dfrac{\mathbf{\Delta }\Phi ^{t+\varepsilon }}{2}\mathbf{V}^{\varepsilon }-%
\mathbf{Hess}\Phi ^{t+\varepsilon }\mathbf{V}^{\varepsilon }
\end{equation}%
is bounded in $B^{0}$ uniformly with respect to $\varepsilon $ for $%
\varepsilon $ small enough.

Taking again scalar product with $\mathbf{V}^{\varepsilon }$ on both sides
of $\left( 55\right) $ and integrating by parts, we obtain%
\begin{equation}
\left\| \mathbf{\nabla V}^{\varepsilon }\right\| _{L^{2}}^{2}+\left\| \frac{%
\left| \mathbf{\nabla }\Phi ^{t+\varepsilon }\right| }{2}\mathbf{V}%
^{\varepsilon }\right\| _{L^{2}}^{2}\leq \left\| \tilde{q}_{\varepsilon
}\right\| _{L^{2}}\left\| \mathbf{V}^{\varepsilon }\right\| _{L^{2}}
\end{equation}%
It follows that $\mathbf{V}^{\varepsilon }$ is uniformly bounded with
respect to $\varepsilon $ in $B_{\Phi ^{t+\varepsilon }}^{1}$ for $%
\varepsilon $ small enough.

Next, observe that 
\begin{equation}
\left( \mathbf{-\Delta +}\frac{\left| \mathbf{\nabla }\Phi ^{t}\right| ^{2}}{%
4}\right) \mathbf{V}^{\varepsilon }=\hat{q}_{\varepsilon }
\end{equation}%
where 
\begin{eqnarray}
\hat{q}_{\varepsilon } &=&\tilde{q}_{\varepsilon }-\frac{\left| \mathbf{%
\nabla }\Phi ^{t+\varepsilon }-\mathbf{\nabla }\Phi ^{t}\right| ^{2}}{4}%
\mathbf{V}^{\varepsilon }+\frac{\left( \mathbf{\nabla }\Phi ^{t+\varepsilon
}-\mathbf{\nabla }\Phi ^{t}\right) \cdot \mathbf{\nabla }\Phi ^{t}}{2}%
\mathbf{V}^{\varepsilon } \\
&=&\tilde{q}_{\varepsilon }-\frac{\varepsilon ^{2}\left| \mathbf{\nabla }%
g\right| ^{2}}{4}\mathbf{V}^{\varepsilon }-\frac{\varepsilon \mathbf{\nabla }%
g\cdot \mathbf{\nabla }\Phi ^{t}}{2}\mathbf{V}^{\varepsilon }
\end{eqnarray}%
is uniformly bounded in $B^{0}$ with respect to $\varepsilon $ for small
enough $\varepsilon .$ Using \ regularity, it follows that for small enough $%
\varepsilon ,$ $\mathbf{V}^{\varepsilon }$ is uniformly bounded in $B_{\Phi
^{t}}^{2}$ with respect to $\varepsilon .$This implies that $\hat{q}%
_{\varepsilon }$ is uniformly bounded in $B_{\Phi ^{t}}^{1}$ for $%
\varepsilon $ small enough. Again, we can continue by a bootstrap argument
to consequently get that for $\varepsilon $ small enough, $\mathbf{V}%
^{\varepsilon }$ is uniformly bounded in $B_{\Phi ^{t}}^{k}$ with respect to 
$\varepsilon $ for any $k.$

Let 
\begin{equation*}
\mathbf{V=w}(t)e^{-\Phi ^{t}/2}.
\end{equation*}%
We have 
\begin{equation}
\left. 
\begin{array}{c}
\left( \mathbf{-\Delta +}\dfrac{\left| \mathbf{\nabla }\Phi ^{t}\right| ^{2}%
}{4}-\dfrac{\mathbf{\Delta }\Phi ^{t}}{2}\right) \mathbf{V}+\mathbf{Hess}%
\Phi ^{t}\mathbf{V} \\ 
=\left[ \mathbf{Hess}g\mathbf{v}(t)+\mathbf{\nabla }g\cdot \mathbf{\nabla v}%
(t)\right] e^{-\Phi ^{t}/2}%
\end{array}%
\right. .
\end{equation}%
Now combining this equation with (51), we obtain

\begin{equation}
\left. 
\begin{array}{c}
\left( \mathbf{-\Delta +}\dfrac{\left| \mathbf{\nabla }\Phi ^{t}\right| ^{2}%
}{4}-\dfrac{\mathbf{\Delta }\Phi ^{t}}{2}\right) \left( \mathbf{V}%
^{\varepsilon }-\mathbf{V}\right) +\mathbf{Hess}\Phi ^{t}\left( \mathbf{V}%
^{\varepsilon }-\mathbf{V}\right) \\ 
=-\left[ \mathbf{Hess}g\mathbf{v}(t)+\mathbf{\nabla }g\cdot \mathbf{\nabla v}%
(t)\right] e^{-\Phi ^{t}/2}+\left[ \mathbf{Hess}g\mathbf{v}(t)+\mathbf{%
\nabla }g\cdot \mathbf{\nabla v}(t)\right] e^{-\Phi ^{t+\varepsilon }/2} \\ 
+\left( \dfrac{\left| \mathbf{\nabla }\Phi ^{t}\right| ^{2}}{4}-\dfrac{%
\left| \mathbf{\nabla }\Phi ^{t+\varepsilon }\right| ^{2}}{4}\right) \mathbf{%
V}^{\varepsilon }-\left( \dfrac{\mathbf{\Delta }\Phi ^{t}}{2}-\dfrac{\mathbf{%
\Delta }\Phi ^{t+\varepsilon }}{2}\right) \mathbf{V}^{\varepsilon } \\ 
+\left( \mathbf{Hess}\Phi ^{t}-\mathbf{Hess}\Phi ^{t+\varepsilon }\right) 
\mathbf{V}^{\varepsilon }.%
\end{array}%
\right. .
\end{equation}
Now let us check that \ for small enough $\varepsilon ,$ the right hand side
of (62) is $\mathcal{O}(\varepsilon )$ in $B^{0}.$

For the first term, we have%
\begin{equation*}
\left. 
\begin{array}{c}
-\left[ \mathbf{Hess}g\mathbf{v}(t)+\mathbf{\nabla }g\cdot \mathbf{\nabla v}%
(t)\right] e^{-\Phi ^{t}/2}+\left[ \mathbf{Hess}g\mathbf{v}(t)+\mathbf{%
\nabla }g\cdot \mathbf{\nabla v}(t)\right] e^{-\Phi ^{t+\varepsilon }/2} \\ 
=\left[ \mathbf{Hess}g\mathbf{v}(t)+\mathbf{\nabla }g\cdot \mathbf{\nabla v}%
(t)\right] e^{-\Phi ^{t}/2}\left( e^{\varepsilon g/2}-1\right) \\ 
\sim \frac{\varepsilon }{2}\left[ \mathbf{Hess}g\mathbf{v}(t)+\mathbf{\nabla 
}g\cdot \mathbf{\nabla v}(t)\right] ge^{-\Phi ^{t}/2}%
\end{array}%
\right.
\end{equation*}%
Thus for $\varepsilon $ small enough%
\begin{equation*}
\left\| -\left[ \mathbf{Hess}g\mathbf{v}(t)+\mathbf{\nabla }g\cdot \mathbf{%
\nabla v}(t)\right] e^{-\Phi ^{t}/2}+\left[ \mathbf{Hess}g\mathbf{v}(t)+%
\mathbf{\nabla }g\cdot \mathbf{\nabla v}(t)\right] e^{-\Phi ^{t+\varepsilon
}/2}\right\| _{B^{0}}\leq C\varepsilon .
\end{equation*}%
For the second term, we have%
\begin{eqnarray*}
&&\left| \dfrac{\left| \mathbf{\nabla }\Phi ^{t}\right| ^{2}}{4}-\dfrac{%
\left| \mathbf{\nabla }\Phi ^{t+\varepsilon }\right| ^{2}}{4}\right| \\
&=&\frac{1}{4}\left( \left| \mathbf{\nabla }\Phi ^{t}\right| +\left| \mathbf{%
\nabla }\Phi ^{t+\varepsilon }\right| \right) \left( \left| \mathbf{\nabla }%
\Phi ^{t}\right| +\left| \mathbf{\nabla }\Phi ^{t+\varepsilon }\right|
\right) \\
&\leq &\frac{\varepsilon }{4}\left| \mathbf{\nabla }g\right| \left( \left| 
\mathbf{\nabla }\Phi ^{t}\right| +\left| \mathbf{\nabla }\Phi
^{t+\varepsilon }\right| \right) \\
&\leq &\frac{\varepsilon }{4}\left| \mathbf{\nabla }g\right| \left( 2\left| 
\mathbf{\nabla }\Phi ^{t}\right| +\varepsilon \left| \mathbf{\nabla }%
g\right| \right)
\end{eqnarray*}

Using now the fact that $\mathbf{V}^{\varepsilon }$ is uniformly bounded in $%
B_{\Phi ^{t}}^{k}$ with respect to $\varepsilon $ for any $k$, we see that
the second term of the right hand side of (62) is $\mathcal{O}(\varepsilon )$
in $B_{\Phi ^{t}}^{0}.$ The last two terms of the right hand side of (62)
are obviously $\mathcal{O}(\varepsilon )$ in $B_{\Phi ^{t}}^{0}.$

From the same regularity argument as above, we get that $\mathbf{V}%
^{\varepsilon }-\mathbf{V}$ is $\mathcal{O}(\varepsilon )$ in $B_{\Phi
^{t}}^{2}$. Again iterating the regularity argument, we obtain that for
small enough $\varepsilon ,$ $\mathbf{V}^{\varepsilon }-\mathbf{V}$ is $%
\mathcal{O}(\varepsilon )$ in $B_{\Phi ^{t}}^{k}$ for every $k.$ We have
proved:

\begin{proposition}
Under the above assumptions on $\Phi $ and $g,$ there exists $T>0$ so that
for each $t\in (0,T),$ $\mathbf{v}^{\varepsilon }(t)$ converges to $\mathbf{w%
}(t)$ in $C^{\infty }.$
\end{proposition}

\begin{remark}
The proposition establishes that $\mathbf{v}(t)$ is differentiable in $t$
and $\dfrac{d}{dt}\mathbf{v}(t)$ is given by the unique $C^{\infty }-$%
solution $\mathbf{w}(t)$ of the system 
\begin{equation}
-\mathbf{\Delta w}(t)+\mathbf{\nabla }\Phi ^{t}\cdot \mathbf{\nabla w}(t)+%
\mathbf{Hess}\Phi ^{t}\mathbf{w}(t)=\mathbf{Hess}g\mathbf{v}(t)-\mathbf{%
\nabla }g\cdot \mathbf{\nabla v}(t).
\end{equation}%
Iterating this argument, we easily get that, $\mathbf{v}(t)$ is smooth in $%
t\in (0,T).$
\end{remark}

Now we are ready for the following:

\section{A Formula for the Taylor Coefficients}

First observe that for an arbitrary suitable function $\ f(t)=f(t,w)$%
\begin{equation}
\frac{\partial }{\partial t}<f(t)>_{t,\Lambda }=<f^{\;\prime
}(t)>_{t,\Lambda }+\mathbf{cov}(f,g).
\end{equation}%
Hence,%
\begin{equation}
\frac{\partial }{\partial t}<f(t)>_{t,\Lambda }=<f^{\;\prime
}(t)>_{t,\Lambda }+<A_{\Phi ^{t}}^{(1)^{-1}}\left( \mathbf{\nabla }f\right)
\cdot \mathbf{\nabla }g>_{t,\Lambda }.
\end{equation}%
Let 
\begin{equation}
A_{g}f:=A_{\Phi ^{t}}^{(1)^{-1}}\left( \mathbf{\nabla }f\right) \cdot 
\mathbf{\nabla }g.
\end{equation}%
Thus,%
\begin{equation}
\frac{\partial }{\partial t}<f(t)>_{t,\Lambda }=<\left( \dfrac{\partial }{%
\partial t}+A_{g}\right) f>_{t,\Lambda }.
\end{equation}%
The linear operator $\dfrac{\partial }{\partial t}+A_{g}$ will be denoted by 
$H_{g}.$

To obtain a formula for the coefficients in the Taylor expansion of 
\begin{equation}
\theta _{\Lambda }(t)=\log \left[ \int dxe^{-\Phi _{\Lambda }^{t}(x)}\right]
,
\end{equation}%
we first the derivatives of $\theta _{\Lambda }(t)$ in terms of $H_{g}$ 
\begin{equation*}
\theta _{\Lambda }^{\prime }(t)=<g>_{t,\Lambda }=<\left( \dfrac{\partial }{%
\partial t}+A_{g}\right) ^{0}g>_{t,\Lambda }=<H_{g}^{0}g>_{t,\Lambda };
\end{equation*}%
\begin{equation*}
\theta _{\Lambda }^{\prime \prime }(t)=\frac{\partial }{\partial t}%
<g>_{t,\Lambda }=<A_{\Phi ^{t}}^{(1)^{-1}}\left( \mathbf{\nabla }g\right)
\cdot \mathbf{\nabla }g>_{t,\Lambda }=<\left( \dfrac{\partial }{\partial t}%
+A_{g}\right) g>_{t,\Lambda };
\end{equation*}%
\begin{eqnarray*}
\theta _{\Lambda }^{\prime \prime \prime }(t) &=&\frac{\partial }{\partial t}%
<A_{\Phi ^{t}}^{(1)^{-1}}\left( \mathbf{\nabla }g\right) \cdot \mathbf{%
\nabla }g>_{t,\Lambda }=<\frac{\partial }{\partial t}\left( A_{\Phi
^{t}}^{(1)^{-1}}\left( \mathbf{\nabla }g\right) \cdot \mathbf{\nabla }%
g\right) >_{t,\Lambda } \\
+ &<&\left( A_{\Phi ^{t}}^{(1)^{-1}}\mathbf{\nabla }\left( A_{\Phi
^{t}}^{(1)^{-1}}\left( \mathbf{\nabla }g\right) \cdot \mathbf{\nabla }%
g\right) \right) \cdot \mathbf{\nabla }g>_{t,\Lambda } \\
&=&<\left( \dfrac{\partial }{\partial t}+A_{g}\right) ^{2}g>_{t,\Lambda }.
\end{eqnarray*}%
By induction it is easy to see that%
\begin{equation*}
\theta _{\Lambda }^{(n)}(t)=<\left( \dfrac{\partial }{\partial t}%
+A_{g}\right) ^{n-1}g>_{t,\Lambda }=<H_{g}^{(n-1)}g>_{t,\Lambda
}\;\;\;\;\;\;\;(\forall n\geq 1)
\end{equation*}%
Next, we propose to find a simpler formula for $\theta _{\Lambda }^{(n)}(t)$
that only involves $A_{g}.$%
\begin{eqnarray*}
H_{g}g &=&A_{\Phi ^{t}}^{(1)^{-1}}\left( \mathbf{\nabla }g\right) \cdot 
\mathbf{\nabla }g \\
&=&A_{g}g
\end{eqnarray*}%
\begin{equation}
H_{g}^{2}g=\frac{\partial }{\partial t}\mathbf{\nabla }f\cdot \mathbf{\nabla 
}g+\left( A_{\Phi ^{t}}^{(1)^{-1}}\mathbf{\nabla }\left( A_{\Phi
^{t}}^{(1)^{-1}}\left( \mathbf{\nabla }g\right) \cdot \mathbf{\nabla }%
g\right) \right) \cdot \mathbf{\nabla }g
\end{equation}%
where $f$ satisfies the equation 
\begin{equation}
\mathbf{\nabla }f=A_{\Phi ^{t}}^{(1)^{-1}}\left( \mathbf{\nabla }g\right) .
\end{equation}%
With $\mathbf{v}(t)=\mathbf{\nabla }f,$ as before, we get 
\begin{equation*}
\frac{\partial }{\partial t}\mathbf{\nabla }f\cdot \mathbf{\nabla }g=A_{\Phi
^{t}}^{(1)^{-1}}\left( \mathbf{Hess}g\mathbf{v}(t)+\mathbf{\nabla }g\cdot 
\mathbf{\nabla v}(t)\right) \cdot \mathbf{\nabla }g
\end{equation*}%
and $H_{g}^{2}$ becomes 
\begin{eqnarray*}
H_{g}^{2}g &=&A_{\Phi ^{t}}^{(1)^{-1}}\left[ \left( \mathbf{Hess}g\mathbf{v}%
(t)+\mathbf{\nabla }g\cdot \mathbf{\nabla v}(t)\right) +\mathbf{\nabla }%
\left( A_{\Phi ^{t}}^{(1)^{-1}}\left( \mathbf{\nabla }g\right) \cdot \mathbf{%
\nabla }g\right) \right] \cdot \mathbf{\nabla }g \\
&=&A_{\Phi ^{t}}^{(1)^{-1}}2\mathbf{\nabla }\left( A_{g}g\right) \cdot 
\mathbf{\nabla }g \\
&=&2A_{g}^{2}g.
\end{eqnarray*}

\begin{proposition}
\textit{If }%
\begin{equation*}
\theta _{\Lambda }(t)=\log \left[ \int dxe^{-\Phi ^{t}(x)}\right]
\end{equation*}%
\textit{where}%
\begin{equation*}
\Phi ^{t}(x)=\Phi _{\Lambda }(x)-tg(x)
\end{equation*}%
\textit{is as above then }$\theta _{\Lambda }^{(n)}(t),$\textit{\ the }$nth-$%
\textit{\ derivative of }$\theta _{\Lambda }(t)$\textit{\ is given by the
formula }%
\begin{equation*}
\theta _{\Lambda }^{\prime }(t)=<g>_{t,\Lambda },
\end{equation*}%
\textit{and for }$n\geq 1$%
\begin{equation*}
\theta _{\Lambda }^{(n)}(t)=\left( n-1\right) !<A_{g}^{n-1}g>_{t,\Lambda }.
\end{equation*}
\end{proposition}

\begin{proof}
We have already established that 
\begin{equation*}
\theta _{\Lambda }^{(n)}(t)=<H_{g}^{n-1}g>_{t,\Lambda }\;\;\;\;for\;n\geq 1.
\end{equation*}%
It then only remains to prove that 
\begin{equation*}
H_{g}^{n-1}g=\left( n-1\right) !A_{g}^{n-1}g\;\;\;\;\;\;for\;n\geq 1.
\end{equation*}%
The result is already established above for $n=1,2,3,.$ By induction, assume
that 
\begin{equation*}
H_{g}^{n-1}g=\left( n-1\right) !A_{g}^{n-1}g\;.
\end{equation*}%
if $n$ is replaced by $\tilde{n}\leq n.$%
\begin{eqnarray*}
H_{g}^{n}g &=&\left( \dfrac{\partial }{\partial t}+A_{g}\right) \left(
\left( n-1\right) !A_{g}^{n-1}g\right) \\
&=&\left( n-1\right) !\left( \dfrac{\partial }{\partial t}%
A_{g}^{n-1}g+A_{g}^{n}g\right) .
\end{eqnarray*}%
Now 
\begin{eqnarray*}
A_{g}^{n-1}g &=&\left[ A_{\Phi ^{t}}^{(1)^{-1}}\mathbf{\nabla }\left(
A_{g}^{n-2}g\right) \right] \cdot \mathbf{\nabla }g\; \\
&=&\mathbf{\nabla }\varphi _{n}\cdot \mathbf{\nabla }g
\end{eqnarray*}%
where%
\begin{equation*}
\mathbf{\nabla }\varphi _{n}=\left[ A_{\Phi ^{t}}^{(1)^{-1}}\mathbf{\nabla }%
\left( A_{g}^{n-2}g\right) \right] .
\end{equation*}%
We obtain,%
\begin{equation*}
\dfrac{\partial }{\partial t}\mathbf{\nabla }\varphi _{n}=A_{\Phi
^{t}}^{(1)^{-1}}\left( \dfrac{\partial }{\partial t}\mathbf{\nabla }%
A_{g}^{n-2}g+\mathbf{Hess}g\mathbf{\nabla }\varphi _{n}+\mathbf{\nabla }%
g\cdot \mathbf{\nabla }\left( \mathbf{\nabla }\varphi _{n}\right) \right) .
\end{equation*}%
We then have 
\begin{eqnarray*}
\dfrac{\partial }{\partial t}A_{g}^{n-1}g &=&\dfrac{\partial }{\partial t}%
\mathbf{\nabla }\varphi _{n}\cdot \mathbf{\nabla }g \\
&=&\left[ A_{\Phi ^{t}}^{(1)^{-1}}\left( \dfrac{\partial }{\partial t}%
\mathbf{\nabla }A_{g}^{n-2}g+\mathbf{Hess}g\mathbf{\nabla }\varphi _{n}+%
\mathbf{\nabla }g\cdot \mathbf{\nabla }\left( \mathbf{\nabla }\varphi
_{n}\right) \right) \right] \cdot \mathbf{\nabla }g \\
&=&\left[ A_{\Phi ^{t}}^{(1)^{-1}}\left( \dfrac{\partial }{\partial t}%
\mathbf{\nabla }A_{g}^{n-2}g+\mathbf{\nabla }\left( \mathbf{\nabla }\varphi
_{n}\cdot \mathbf{\nabla }g\right) \right) \right] \cdot \mathbf{\nabla }g \\
&=&A_{g}\left[ \dfrac{\partial }{\partial t}A_{g}^{n-2}g+A_{g}\left(
A_{g}^{n-2}g\right) \right] \\
&=&A_{g}H_{g}\left( A_{g}^{n-2}g\right) . \\
&=&A_{g}H_{g}\left( \frac{1}{\left( n-2\right) !}H_{g}^{(n-2)}g\right)
\;\;\;\;\;\;\;(\text{from the induction hypothesis}) \\
&=&\frac{1}{\left( n-2\right) !}A_{g}H_{g}^{(n-1)}g\; \\
&=&\frac{1}{\left( n-2\right) !}A_{g}\left( \left( n-1\right)
!A_{g}^{n-1}g\right) \;\;\;\;\;(\text{still by the induction hypothesis})\;\
\ \ \  \\
&=&(n-1)A_{g}^{n}g.
\end{eqnarray*}%
Thus,%
\begin{eqnarray*}
H_{g}^{n}g &=&\left( n-1\right) !\left( n-1+1\right) A_{g}^{n}g \\
&=&n!A_{g}^{n}g
\end{eqnarray*}
\end{proof}

\begin{proposition}
\textit{If }$g(0)=0,$ \textit{then }the formula 
\begin{equation*}
\theta _{\Lambda }^{(n)}(t)=\left( n-1\right) !<A_{g}^{n-1}g>_{t,\Lambda
},\;\;\;n\geq 2
\end{equation*}%
still holds if we no longer require $\Psi $ and $g$ to be compactly
supported in $\mathbb{R}^{\Lambda }.$
\end{proposition}

\begin{proof}
As in [8], consider the family cutoff functions 
\begin{equation}
\chi =\chi _{\varepsilon }
\end{equation}%
$(\varepsilon \in \lbrack 0,1])$ in $\mathcal{C}_{o}^{\infty }(\mathbb{R})$
with value in $[0,1]$ such that%
\begin{equation*}
\left\{ 
\begin{array}{c}
\chi =1\text{ \ \ \ \ \ \ \ \ \ \ \ \ \ \ \ \ for }\left| t\right| \leq
\varepsilon ^{-1}\text{ } \\ 
\left| \chi ^{(k)}(t)\right| \leq C_{k}\dfrac{\varepsilon }{\left| t\right|
^{k}}\text{ \ \ \ \ \ \ \ \ \ \ \ \ \ \ \ \ for }k\in \mathbb{N\ }\text{\ \
\ \ \ \ \ \ \ \ \ \ \ \ \ \ \ \ \ \ \ \ \ }%
\end{array}%
\right.
\end{equation*}%
We could take for instance%
\begin{equation*}
\chi _{\varepsilon }(t)=f(\varepsilon \ln \left| t\right| )
\end{equation*}%
for a suitable $f$.

We then introduce%
\begin{equation}
\Psi _{\varepsilon }(x)=\chi _{\varepsilon }(\left| x\right| )\Psi ,\ \ \ \
\ \ \ \ \ x\in \mathbb{R}^{\Lambda }
\end{equation}%
and 
\begin{equation}
g_{\varepsilon }(x)=\chi _{\varepsilon }(\left| x\right| )g\text{ \ \ \ \ \
\ \ \ \ \ \ \ }x\in \mathbb{R}^{\Gamma }\text{\ }
\end{equation}%
One can check that both $\Psi _{\varepsilon }(x)$ and $g_{\varepsilon }(x)$
satisfies the assumptions made above on $\Psi $ and $g.$ Now consider the
equation 
\begin{equation}
-\mathbf{\Delta }f_{\varepsilon }+\mathbf{\nabla }\Phi _{\varepsilon
}^{t}\cdot \mathbf{\nabla }f_{\varepsilon }=g_{\varepsilon }-<g_{\varepsilon
}>_{t,\Lambda .}
\end{equation}%
which implies 
\begin{equation}
\left( -\mathbf{\Delta }+\mathbf{\nabla }\Phi _{\varepsilon }^{t}\cdot 
\mathbf{\nabla }\right) \otimes \mathbf{v}_{\varepsilon }+\mathbf{Hess}\Phi
_{\varepsilon }^{t}\mathbf{v}_{\varepsilon }=\mathbf{\nabla }g_{\varepsilon }
\end{equation}%
where 
\begin{equation*}
\mathbf{v}_{\varepsilon }\mathbf{=\nabla }f_{\varepsilon }
\end{equation*}%
It was proved in [8] that $\mathbf{v}_{\varepsilon }=A_{\Phi ^{t}}^{(1)^{-1}}%
\mathbf{\nabla }g_{\varepsilon }\;$converges in $C^{\infty }$ to $A_{\Phi
^{t}}^{(1)^{-1}}\mathbf{\nabla }g$ as $\varepsilon \rightarrow 0.$
\end{proof}

\begin{proposition}
Let%
\begin{equation*}
P_{\Lambda }(t)=\dfrac{\theta _{\Lambda }(t)}{\left| \Lambda \right| }
\end{equation*}%
be the finite volume Pressure. \newline
Denote by $a_{n}\;(n\geq 2)$ the $nth$ Taylor coefficient. We have 
\begin{equation*}
a_{n}=\frac{<A_{g}^{n-1}g>_{\Lambda }}{n\left| \Lambda \right| }
\end{equation*}
\end{proposition}

\begin{remark}
This formula for $a_{n}$ gives a direction towards proving the analyticity
of the pressure in the thermodynamic limit. In fact one only needs to
provide a suitable $C^{n}$ estimate for $<A_{g}^{n-1}g>_{\Lambda }.$
\end{remark}

\section{Some Consequences of the Formula for $nth-$Derivative of the
Pressure.}

In the following, we shall additionally assume that 
\begin{equation*}
\left. 
\begin{array}{c}
\mathbf{\nabla }g(0)=0,\;\;\;\;\;\text{and} \\ 
\mathbf{\nabla }\Phi _{\Lambda }^{t}(0)=0\text{ \ \ \ for all }t\in \lbrack
0,T).%
\end{array}%
\right.
\end{equation*}%
When $n=1,$ we recall that $A_{g}^{0}g=g$, 
\begin{equation*}
\theta _{\Lambda }^{\prime }(t)=<g>_{t,\Lambda }
\end{equation*}%
and if 
\begin{equation*}
\mathbf{v}(t)=\mathbf{\nabla }f=A_{\Phi ^{t}}^{(1)^{-1}}\mathbf{\nabla }g,
\end{equation*}%
then we have 
\begin{equation*}
\left( -\mathbf{\Delta }+\mathbf{\nabla }\Phi _{\Lambda }^{t}\cdot \mathbf{%
\nabla }\right) \otimes \mathbf{v}(t)+\mathbf{Hess}\Phi _{\Lambda }^{t}%
\mathbf{v}(t)=\mathbf{\nabla }g
\end{equation*}%
Again the tensor notation means that $\left( -\mathbf{\Delta }+\mathbf{%
\nabla }\Phi _{\Lambda }^{t}\cdot \mathbf{\nabla }\right) $ acts diagonally
on the components of $\mathbf{v}(t).$

As in [8] $\mathbf{v}(t)$ is a solution of the equation%
\begin{equation}
g=<g>_{t,\Lambda }+\mathbf{v}(t)\cdot \mathbf{\nabla }\Phi _{\Lambda
}^{t}-div\mathbf{v}(t).
\end{equation}%
Using the assumptions above, we have 
\begin{eqnarray*}
\theta _{\Lambda }^{\prime }(t) &=&<g>_{t,\Lambda } \\
&=&div\mathbf{v}(t)(0).
\end{eqnarray*}%
Similarly, the formula 
\begin{equation*}
\theta _{\Lambda }^{(n)}(t)=\left( n-1\right) !<A_{g}^{n-1}g>_{t,\Lambda },
\end{equation*}%
implies that 
\begin{equation*}
\theta _{\Lambda }^{(n)}(t)=\left( n-1\right) !div\mathbf{v}_{n}(t)(0),
\end{equation*}%
where 
\begin{equation*}
\mathbf{v}_{n}(t)=A_{\Phi ^{t}}^{(1)^{-1}}\mathbf{\nabla }\left(
A_{g}^{n-1}g\right) .
\end{equation*}%
\textbf{Acknowledgements.} I would like to thank Professor. Haru Pinson and
Professor. Tom Kennedy for accepting to discuss with me the ideas developed
in this paper. I also would like to acknowledge\ Professor Bruno
Nachtergaele for his constructive suggestions and all members of the
mathematical physics group at the University of Arizona for their support.

\end{document}